\documentclass[envcountsame, envcountsect, runningheads, a4paper]{llncs}
\usepackage[dvipsnames]{xcolor}
\usepackage{amsmath,amsfonts,amssymb}
\usepackage[utf8]{inputenc}
\usepackage{graphicx}
\usepackage{tikz}
\usepackage[bookmarks,unicode]{hyperref}
\usepackage[backgroundcolor=magenta!10!white]{todonotes}
\usepackage{setspace}
\usepackage{microtype}
\usepackage{enumitem}
\usepackage{ifthen}
\usepackage{caption}
\usepackage{microtype}
\usepackage[space]{cite}
\usepackage{numprint}
\usepackage{cleveref}
\usepackage{multirow}
\usepackage{subcaption}
\usepackage{placeins}
\usepackage[linesnumbered,ruled,vlined]{algorithm2e}

\spnewtheorem{notation}[theorem]{Notation}{\bfseries}{\upshape}
\spnewtheorem{setting}[theorem]{Setting}{\bfseries}{\upshape}
\usepackage{thmtools}

\usetikzlibrary{arrows,
	automata,
	calc,
	decorations,
	decorations.pathreplacing,
	positioning,
	shapes}
\tikzset{res/.style={ellipse,draw,minimum height=0.5cm,minimum width=0.8cm}}

\setenumerate[1]{label={(\arabic*)}} 
\numberwithin{equation}{section}
\npdecimalsign{.}
\npthousandsep{,}

\newcommand{\A}{\mathcal{A}}

\newcommand{\lang}{\mathcal{L}}

\renewcommand{\inst}{{\operatorname{inst}}}
\newcommand{\mdp}{\mathcal{C}}

\newcommand{\Act}{\operatorname{Act}}

\newcommand{\last}{\operatorname{last}}
\newcommand{\Paths}{\operatorname{Paths}}

\newcommand{\Cyl}{\operatorname{Cyl}}

\renewcommand{\Pr}{\mathrm{Pr}}
\newcommand{\fin}{{\operatorname{fin}}}

\newcommand{\Pref}{\operatorname{Pref}}
\newcommand{\expcost}{\operatorname{expcost}}
\newcommand{\pexpcost}{\operatorname{pexpcost}}
\newcommand{\maxcost}{\operatorname{maxcost}}

\newcommand{\ExpRew}{\operatorname{ExpRew}}

\newcommand{\pp}{\mathtt{PP}}
\newcommand{\ptime}{\mathtt{P}}
\newcommand{\np}{\mathtt{NP}}
\newcommand{\conp}{\mathtt{coNP}}
\newcommand{\Goal}{\lozenge E}
\newcommand{\goal}{\lozenge \mathit{error}}
\newcommand{\goalst}{\mathit{error}}
\newcommand{\failst}{\mathit{safe}}
\newcommand\sto{\mathrel{\overset{\makebox[0pt]{\mbox{\normalfont\tiny\sffamily $s'$}}}{\mapsto}}}
\newcommand\iotato{\mathrel{\overset{\makebox[0pt]{\mbox{\normalfont\tiny\sffamily $s_0$}}}{\mapsto}}}
\newenvironment{proofsketch}{\proof}{\endproof}
\newif \iflongversion
\longversionfalse

\iflongversion
\newcommand{\citeapp}[1]{\Cref{#1}}
\else
\newcommand{\citeapp}[1]{the full version~\cite{FunkeJB19}}
\fi

\iflongversion

\else

\fi

\begin{document}

\title{Probabilistic causes in Markov chains}

%
%
\author{Christel Baier \and Florian Funke \and Simon Jantsch\and\\ Jakob Piribauer \and Robin Ziemek}
\authorrunning{Baier et al.}
%
\institute{Technische Universität Dresden\thanks{This work was funded by DFG grant 389792660 as part of TRR~248, the Cluster of Excellence EXC 2050/1 (CeTI, project ID 390696704, as part of Germany’s Excellence Strategy), DFG-projects BA-1679/11-1 and BA-1679/12-1, and the Research Training Group QuantLA (GRK 1763).} \\
	\email{\{christel.baier, florian.funke, simon.jantsch,\\ jakob.piribauer, robin.ziemek\}@tu-dresden.de}}
\maketitle
\begin{abstract}
	The paper studies a probabilistic notion of causes in Markov chains that relies on the \emph{counterfactuality} principle and the \emph{probability-raising} property. This notion is motivated by the use of causes for monitoring purposes where the aim is to detect faulty or undesired behaviours \emph{before} they actually occur. A cause is a set of finite executions of the system \emph{after which} the probability of the effect exceeds a given threshold. We introduce multiple types of costs that capture the consump-tion of resources from different perspectives, and study the complexity of computing cost-minimal causes

\end{abstract}

\section{Introduction}

The study of cause-effect relationships in formal systems has received considerable attention over the past 25 years.
Notions of causality have been proposed within various models, including structural equation models \cite{Pearl09,HalpernP2001,Halpern15}, temporal logics in Kripke structures \cite{ChocklerHK2008,BeerBCOT09} and Markov chains \cite{KleinbergM2009, KleinbergM2010}, and application areas have been identified in abundance, ranging from finance \cite{Kleinberg2011} to medicine \cite{KleinbergH11} to aeronautics \cite{IbrahimPTESA20}. 
These approaches form an increasingly powerful toolkit aimed at explaining \emph{why} an observable phenomenon (the effect) has happened, and which previous events (the causes) are logically linked to its occurrence. 
As such, causality plays a fundamental building block in determining moral responsibility \cite{ChocklerH04, BrahamvanHees2012} or legal accountability \cite{FeigenbaumHJJWW11}, and ultimately fosters user acceptance through an increased level of transparency \cite{Miller17}.

Despite the variety of models, application areas, and involved disciplines, all approaches essentially rely on (one of) two central paradigms that dictate how causes are linked to their effects: the \emph{counterfactuality} principle and the \emph{probability-raising} property. Counterfactual reasoning prescribes that an effect would not have happened if the cause had not occurred. 
Probability-raising states that the probability of the effect is higher whenever the cause has been observed.

The contribution of this paper is twofold: First, we define a novel notion of \emph{cause} for $\omega$-regular properties in stochastic operational models. Second, we study the complexity of computing optimal causes for cost mechanisms motivated by monitoring applications.

The causes presented in this paper combine the two prevailing causality paradigms mentioned above into a single concept. More specifically, a $p$-cause for an $\omega$-regular property $\lang$ in a discrete-time Markov chain is a set of finite executions $\pi$ of the system such that the probability that $\lang$ occurs after executing $\pi$ is at least $p$, where $p$ is typically larger than the overall probability of $\lang$. 
The counterfactuality principle is invoked through the additional requirement that almost every execution exhibiting the event $\lang$ contains a finite prefix which is a member of the $p$-cause.
This condition makes our approach amenable to the needs of monitoring a system at runtime. 

Imagine a critical event that the system should avoid (e.g., a fully automated drone crashing onto the ground), and assume that a $p$-cause for this event is known (e.g., physical specifications foreshadowing a crash).
Typically, the probability threshold $p$ -- which can be thought of as the sensitivity of the monitor -- should be lower if the criticality of the event is higher. 
As the system is running, as soon as the execution seen so far is part of the $p$-cause, the monitor can trigger an alarm and suitable countermeasures can be taken (e.g., manual control instead of automated behavior). 
As such, our approach can be \emph{preventive} in nature.

The monitoring application outlined above suggests computing a $p$-cause from the model before the system is put to use. 
However, multiple $p$-causes may exist for the same property, which raises the question which one to choose. 
Cyber-physical systems consume time, energy and other resources, which are often subject to budget restrictions.
Furthermore, the intended countermeasures may incur different costs depending on the system state.
Such costs can be modelled using state weights in the Markov chain, which induce weight functions on the finite executions either in an accumulative (total resource consumption) or instantaneous (current consumption intensity) fashion. 
On top of this model, we present three cost mechanisms for causes: (1) The \emph{expected cost} measures the expected resource consumption until the monitor triggers an alarm or reaches a safe state, (2) the \emph{partial expected cost} measures the expected consumption where executions reaching a safe state do not incur any cost, and (3) the \emph{maximal cost} measures the maximal consumption that can occur until an alarm is triggered.

\begin{figure}[t]
	\centering
	{\def\arraystretch{1.1}\tabcolsep=5pt
		\begin{tabular}{|c|c|c|c|}
			\hline
			& \multirow{2}{*}{$\expcost$} & \multirow{2}{*}{$\pexpcost$} &  \multirow{2}{*}{$\maxcost$} \\
			&&& \\
			\hline
			non-negative weights & \multirow{2}{*}{in $\ptime$ (\ref{thm: complexity expcost})}  &pseudo-polyn. (\ref{thm: expcost0 minimal path}) & \multirow{2}{*}{in $\ptime$ (\ref{thm: maxcost})}\\
			accumulated & & $\pp$-hard (\ref{thm: expcost0 PP-hard})   & \\\hline
			arbitrary weights & \multirow{2}{*}{in $\ptime$ (\ref{thm: complexity expcost})} & \multirow{2}{*}{$\pp$-hard (\ref{thm: expcost0 PP-hard})} & pseudo-polyn. \multirow{2}{*}{(\ref{thm: maxcost})}\\
			accumulated &&& in $\np \cap \conp$  \multirow{2}{*}{\hspace{7mm}} \\\hline
			arbitrary weights& \multirow{2}{*}{in $\ptime$ (\ref{thm: instcost}) }& \multirow{2}{*}{in $\ptime$ (\ref{thm: instcost})}&  \multirow{2}{*}{in $\ptime$ (\ref{thm: instcost})}\\
			instantaneous &&& \\\hline
	\end{tabular}}
	\caption{Summary of complexity results for different kinds of cost.}
	\label{tab: main table}
\end{figure}

\Cref{tab: main table} summarizes our results regarding the complexity of computing cost-minimal $p$-causes for the different combinations of weight type and cost mechanism. 
To obtain these results we utilize a web of connections to the rich landscape of computational problems for discrete-time Markovian models. 
More precisely, the results for the expected cost rely on connections to the \emph{stochastic shortest path problem} (SSP) studied in \cite{BT91}. The pseudo-polynomial algorithm for partial expected costs on non-negative, accumulated weights uses \emph{partial expectations} in Markov decision processes \cite{fossacs2019}. 
The $\pp$-hardness result is proved by reduction from the \emph{cost problem} for acyclic Markov chains stated in \cite{HaaseK14}. 
The pseudo-polynomial algorithm for the maximal cost on arbitrary, accumulated weights applies insights from \emph{total-payoff games} \cite{BGHM2015, Chatterjee2017}.

Full proofs missing in the main document can be found in the appendix.

\subsubsection{Related Work.}
The structural model approach to actual causality \cite{HalpernP2001} has sparked notions of causality in formal verification \cite{ChocklerHK2008, BeerBCOT09}. 
The complexity of computing actual causes has been studied in~\cite{EiterL2004, EiterL2006}. 
A probabilistic extension of this framework has been proposed in \cite{Fenton-Glynn2016}. 
Recent work on checking and inferring actual causes is given in \cite{IbrahimP20}, and an application-oriented framework for it is presented in \cite{IbrahimPTESA20}.
The work \cite{KleinbergM2009} builds a framework for actual causality in Markov chains and applies it to infer causal relationships in data sets. 
It was later extended to continuous time data~\cite{Kleinberg2011} and to token causality~\cite{KleinbergM2010} and has been refined using new measures for the significance of actual and token causes \cite{HuangK15,ZhengK2017}.

A logic for probabilistic causal reasoning is given in \cite{VennekensDB09} in combination with logical programming. 
The work \cite{VennekensBD10} compares this approach to Pearl's theory of causality involving  Bayesian networks \cite{Pearl09}. 
The CP-logic of \cite{VennekensDB09} is close to the representation of causal mechanisms of \cite{DVJ2013}. 
The probability-raising principle goes back to Reichenbach \cite{Reichenbach56}.
It has been identified as a key ingredient to causality in various philosophical accounts, see e.g. \cite{Eells91}. 

Monitoring $\omega$-regular properties in stochastic systems modeled as Hidden Markov Chains (HMCs) was studied in \cite{SistlaS08,GondiPS09} and has recently been revived \cite{EsparzaKKW20}.
The trade-off between accuracy and overhead in runtime verification has been studied in \cite{RV11,RV12,RV13}.
In particular \cite{RV11} uses HMCs to estimate how likely each monitor instance is to violate a temporal property.
Monitoring the evolution of finite executions has also been investigated in the context of statistical model checking of LTL properties \cite{DacaHKP16}. How randomization can improve monitors for non-probabilistic systems has been examined in \cite{ChadhaSV09}. 
The safety level of \cite{FaranK2018} measures which portion of a language admits bad prefixes, in the sense classically used for safety languages.

\section{Preliminaries}\label{sec:prelims}

\subsubsection{Markov chains.} A \emph{discrete-time Markov chain} (DTMC) $M$ is a tuple $(S, s_0, \mathbf{P})$, where $S$ is a finite set of \emph{states}, $s_0\in S$ is the \emph{initial state}, and $\mathbf{P} \colon S \times S \to [0,1]$ is the \emph{transition probability function} where we require $\sum_{s' \in S} \mathbf{P}(s, s') = 1$ for all $s \in S$. For algorithmic problems all transition probabilities are assumed to be rational.
A \emph{finite path} $\hat\pi$ in $M$ is a sequence $s_0s_1\ldots s_n$ of states such that $\mathbf{P}(s_i, s_{i+1}) > 0$ for all $0 \leq i \leq n-1$. Let $\last(s_0\ldots s_n) = s_n$. Similarly one defines the notion of an \emph{infinite path} $\pi$. Let $\Paths_\fin(M)$ and $\Paths(M)$ be the set of finite and infinite paths. 
The set of prefixes of a path $\pi$ is denoted by $\Pref(\pi)$.
The \emph{cylinder set} of a finite path $\hat\pi$ is $\Cyl(\hat\pi) =\{\pi\in\ \Paths(M)\mid \hat\pi\in\Pref(\pi)\}$. We consider $\Paths(M)$ as a probability space whose $\sigma$-algebra is generated by such cylinder sets and whose probability measure is induced by $\Pr(\Cyl(s_0\ldots s_n)) = \mathbf{P}(s_0, s_{1})\cdot \ldots \cdot \mathbf{P}(s_{n-1}, s_n) $ (see \cite[Chapter 10]{BaierK2008} for more details). 

For an $\omega$-regular language $\lang \subseteq S^{\omega}$ let $\Paths_M(\lang) = \Paths(M) \cap \lang$. The probability of $\lang$ in $M$ is defined as $\Pr_M(\lang) = \Pr(\Paths_M(\lang))$. Given a state $s\in S$, let $\Pr_{M,s}(\lang) = \Pr_{M_s}(\lang)$, where $M_s$ is the DTMC obtained from $M$ by replacing the initial state $s_0$ with $s$.
If $M$ is clear from the context, we omit the subscript.
For a finite path $\hat\pi \in \Paths_\fin(M)$, define the conditional probability 
\[ \Pr_M(\lang \mid \hat\pi) =  \frac{\Pr_M \left(\Paths_M(\lang) \cap \Cyl(\hat\pi) \right)}{\Pr_M(\Cyl(\hat\pi))}.\]	

Given $E \subseteq S$, let $\Goal = \{ s_0s_1\ldots \in \Paths(M)\mid \exists i \geq 0. \; s_i\in E\}$. 
For such reachability properties we have $\Pr_M(\Goal \mid s_0\ldots s_n) = \Pr_{M, s_n}(\Goal)$ for any $s_0\ldots s_n\in \Paths_\fin(M)$. We assume $\Pr_{s_0}(\lozenge s) > 0$ all states $s \in S$.
Furthermore, define a \emph{weight function} on $M$ as a map $c: S \to \mathbb{Q}$.
We typically use it to induce a weight function $c: \Paths_\fin(M) \to \mathbb{Q}$ (denoted by the same letter) by accumulation, i.e., $c(s_0 \cdots s_n) = \sum_{i=0}^n c(s_i)$.
Finally, a set $\Pi \subseteq \Paths_\fin(M)$ is called \emph{prefix-free} if for every $\hat\pi\in \Pi$ we have $\Pi \cap \Pref(\hat\pi) = \{\hat\pi\}$. 

\subsubsection{Markov decision processes.} A \emph{Markov decision process} (MDP) $\mathcal{M}$ is a tuple $(S, \Act, s_0, \mathbf{P})$, where $S$ is a finite set of \emph{states}, $\Act$ is a finite set of actions, $s_0$ is the \emph{initial state}, and $\mathbf{P} \colon S \times \Act \times S \to [0,1]$ is the \emph{transition probability function} such that for all states $s \in S$ and actions $\alpha \in \Act$ we have $\sum_{s' \in S} \mathbf{P}(s, \alpha,s') \in \{0,1 \}$.
An action $\alpha$ is \emph{enabled} in state $s \in S$ if $\sum_{s' \in S} \mathbf{P}(s, \alpha,s')=1$ and we define $\Act(s) = \{\alpha \mid \alpha \text{ is enabled in } s\}$.
We require $\Act(s) \neq \emptyset$ for all states $s \in S$.

An infinite \emph{path} in $\mathcal{M}$ is an infinite sequence $\pi=s_0 \alpha_1s_1 \alpha_2s_2 \dots \in(S \times \Act)^\omega$ such that for all $i \geq 0$ we have $\mathbf{P}(s_i, \alpha_{i+1},s_{i+1})>0$. Any finite prefix of $\pi$ that ends in a state is a finite path.
A \emph{scheduler} $\mathfrak{S}$ is a function that maps a finite path $s_0 \alpha_1s_1  \dots s_n$ to an enabled  action $\alpha \in \Act(s_n)$. Therefore it resolves the nondeterminism of the MDP and induces a (potentially infinite) Markov chain $\mathcal{M}_{\mathfrak{S}}$. If the chosen action only depends on the last state of the path, i.e., $\mathfrak{S}(s_0 \alpha_1s_1  \dots s_n) = \mathfrak{S}(s_n)$, then the scheduler is called \emph{memoryless} and naturally induces a finite DTMC. For more details on DTMCs and MDPs we refer to \cite{BaierK2008}.

\section{Causes}

This section introduces a notion of \emph{cause} for $\omega$-regular properties in Markov chains.
For the rest of this section we fix a DTMC $M$ with state space $S$, an $\omega$-regular language $\lang$ over the alphabet $S$ and a threshold $p \in (0,1]$.

\begin{definition}[$p$-critical prefix] \label{def: p good bad path prefix}
	A finite path $\hat\pi$ is a \emph{$p$-critical prefix} for $\lang$ if $\Pr(\lang \mid  \hat\pi) \geq p$.
\end{definition}

\begin{definition}[$p$-cause] \label{def:p-cause}
	A $p$-cause for $\lang$ in $M$  is a prefix-free set of finite paths $\Pi \subseteq \Paths_\fin(M)$ such that
	\begin{enumerate}
		\item almost every $\pi \in \Paths_M(\lang)$ has a prefix $\hat\pi \in \Pi$, and
		\item every $\hat\pi \in \Pi$ is a $p$-critical prefix for $\lang$.
	\end{enumerate}
\end{definition}

Note that condition (1) and (2) are in the spirit of completeness and soundness as used in \cite{CF2015}.
The first condition is our invocation of the counterfactuality principle: Almost every occurrence of the effect (for example, reaching a target set) is preceded by an element in the cause.
If the threshold is chosen such that $p > \Pr_{s_0}(\lang)$, then the second condition reflects the probability-raising principle in that seeing an element of $\Pi$ implies that the probability of the effect $\lang$ has increased over the course of the execution.
For monitoring purposes as described in the introduction it would be misleading to choose $p$ below $\Pr_{s_0}(\lang)$ as this could instantly trigger an alarm before the system is put to use.
Also $p$ should not be too close to $1$ as this may result in an alarm being triggered too late.

If $\lang$ coincides with a reachability property one could equivalently remove the \emph{almost} from (1) of \Cref{def:p-cause}. In general, however, ignoring paths with probability zero is necessary to guarantee the existence of $p$-causes for all $p$.

\begin{example}
	Consider the DTMC $M $ depicted in \Cref{fig: simple example}. For $p = 3/4$, a possible $p$-cause for $\lang = \goal$ in $M$ is given by the set $\Pi_1 = \{s t, s u\}$ since both $t$ and $u$ reach $\goalst$ with probability greater or equal than $p$. The sets $\Theta_1 = \{s t, s u, s t u\}$ and $\Theta_2 = \{s t \goalst, s u\}$ are not $p$-causes: $\Theta_1$ is not prefix-free and for $\Theta_2$ the path $s t u \goalst$ has no prefix in $\Theta_2$. Another $p$-cause is $\Pi_2 =\{s t \goalst, s u, s t u\}$.
\end{example}

\begin{example}
	It can happen that there does not exist any finite $p$-cause. Consider \Cref{fig: infinite} and $p = 1/2$. Since $\Pr_{s}(\lozenge \goalst) < p$, the singleton $\{s \}$ is not a $p$-cause. Thus, for every $n \geq 0$ either $s^nt$ or $s^nt\goalst$ is contained in any $p$-cause, which must therefore be infinite. There may also exist non-regular $p$-causes (as languages of finite words over $S$). For example, for $A= \{n \in \mathbb{N} \mid n \text{ prime} \}$ the $p$-cause $\Pi_A = \{ s_0^n t \mid n \in A \} \cup \{ s_0^m t\goalst \mid m \notin A \}$ is non-regular.

\begin{figure}[t]
	\centering
	\begin{minipage}{0.43\textwidth}
		\centering
    \resizebox{\textwidth}{!}{
		\begin{tikzpicture}[->,>=stealth',shorten >=1pt,auto,node distance=0.5cm, semithick]
			\node[scale=0.8, state] (s0) {$s$};
			\node[scale=0.8, state] (s1) [right = 1.5 of s0] {$t$};
			\node[scale=0.8, state] (s2) [below =0.7 of s1] {$u$};
			\node[scale=0.8, state, accepting] (g) [right =1.5 of s1] {$\goalst$};
			\node[scale=0.8, state] (f) [right = 1.58 of s2] {$\failst$};
			
			\draw[<-] (s0) --++(-0.55,0.55);
			\draw (s0) -- (s1) node[pos=0.5,scale=0.8] {$1/2$};
			\draw (s0) -- (s2) node[pos=0.5,scale=0.8] {$1/2$};
			\draw (s1) -- (s2) node[pos=0.5,scale=0.8] {$7/8$};
			\draw (s1) -- (g) node[pos=0.5,scale=0.8] {$1/8$};
			\draw (s2) -- (g) node[pos=0.5,scale=0.8] {$3/4$};
			\draw (s2) -- (f) node[pos=0.5,scale=0.8] {$1/4$};		
		\end{tikzpicture}}
		\caption{Example DTMC $M$} \label{fig: simple example}
	\end{minipage}\hfill
	\begin{minipage}{0.54\textwidth}
		\centering
    \resizebox{0.85\textwidth}{!}{
		\begin{tikzpicture}[->,>=stealth',shorten >=1pt,auto,node distance=0.5cm, semithick]
			\node[scale=0.8, state] (s0) {$s$};
			\node[scale=0.8, state] (s1) [right = 1.5 of s0]{$t$};
			\node[scale=0.8, state, accepting] (g) [right = 1.5 of s1]{$\goalst$};
			\node[scale=0.8, state] (f) [below = 0.7 of s1]{$\failst$};
			
			\draw[<-] (s0) -- ++(-0.55,0.55);
			\path (s0) edge[loop below] node[scale=0.8] {$1/2$} (s0);
			\draw (s0) -- (s1) node[pos=0.5,scale=0.8] {$1/4$};
			\draw (s0) -- (f) node[pos=0.5,scale=0.8, left, yshift=-3mm] {$1/4$};
			\draw (s1) -- (f) node[pos=0.5,scale=0.8] {$1/10$};
			\draw (s1) -- (g) node[pos=0.5,scale=0.8] {$9/10$};
		\end{tikzpicture}}
		\caption{Infinite and non-regular $1/2$-causes} \label{fig: infinite}
	\end{minipage}
\end{figure}
\end{example}

\begin{remark}[Reduction to reachability properties]\label{rem: reachability properties}
	Let $\mathcal{A}$ be a deterministic Rabin automaton for $\lang$ and consider the product Markov chain $M \otimes \A$ as in \cite[Section 10.3]{BaierK2008}.
	For any finite path $\hat\pi=s_0\ldots s_n \in \Paths_\fin(M)$ there is a unique path $a(\hat\pi) = (s_0,q_1) (s_1, q_2)\ldots (s_n,q_{n+1}) \in \Paths_\fin(M \otimes \A)$ whose projection onto the first factor is $\hat\pi$.
	Under this correspondence, a bottom strongly connected component (BSCC) of $M \otimes \A$ is either \emph{accepting} or \emph{rejecting}, meaning that for every finite
	path reaching this BSCC the corresponding path $\hat\pi$ in $M$ satisfies $\Pr_M(\lang\mid\hat\pi) =1$, or respectively, $\Pr_M(\lang\mid\hat\pi) =0$~\cite[Section 10.3]{BaierK2008}.
	This readily implies that almost every $\pi \in \Paths_M(\lang)$ has a $1$-critical prefix and that, therefore, $p$-causes exist for any $p$.
	
	Moreover, if $U$ is the union of all accepting BSCCs in $M \otimes \A$, then 	\begin{equation}\label{eq:reduction}
	\Pr_{M}(\lang \mid \hat\pi) = \Pr_{M \otimes \A}\big(\lozenge U \mid a(\hat\pi) \big)
	\end{equation}
	holds for all finite paths $\hat\pi$ of $M$~\cite[Theorem 10.56]{BaierK2008}. Hence every $p$-cause $\Pi_1$ for $\lang$ in $M$ induces a $p$-cause $\Pi_2$ for $\lozenge U$ in $M \otimes \A$ by taking $\Pi_2 = \{a(\hat\pi) \mid \hat\pi \in \Pi_1\}$.
	Vice versa, given a $p$-cause $\Pi_2$ for $\lozenge U$ in $M \otimes \A$, then the set of projections of paths in $\Pi_2$ onto their first component is a $p$-cause for $\lang$ in $M$. 
	In summary, the reduction of $\omega$-regular properties on $M$ to reachability properties on the product $M \otimes \A$ also induces a reduction on the level of causes.
\end{remark}

	 \Cref{rem: reachability properties} motivates us to focus on reachability properties henceforth. 
	To apply the algorithms presented in \Cref{sec:costs} to specifications given in richer formalisms such as LTL, one would first have to  apply the reduction to reachability given above, which increases the worst-case complexity exponentially. 
	
	In order to align the exposition with the monitoring application we are targeting, we will consider the target set as representing an erroneous behavior that is to be avoided. 
	After collapsing the target set, we may assume that there is a unique state $\goalst \in S$, so $\lang = \goal$ is the language we are interested in. 
	Further, we collapse all states from which $\goalst$ is not reachable to a unique state $\failst \in S$ with the property $\Pr_{\failst}(\goal)=0$.
	After this pre-processing, we have $\Pr_{s_0}(\lozenge \{\goalst, \failst\}) = 1$.
	Define the set 
	\[S_p := \{s \in S \mid \Pr_s(\goal) \geq p\}\] 
	of all acceptable final states for $p$-critical prefixes.
	This set is never empty as $\goalst \in S_p$ for all $p \in (0,1]$.

	There is a partial order on the set of $p$-causes defined as follows: $\Pi \preceq \Phi$ if and only if for all $\phi \in \Phi$ there exists $\pi \in \Pi$ such that  $\pi \in \Pref(\phi)$.
	The reflexivity and transitivity are straightforward, and the antisymmetry follows from the fact that $p$-causes are prefix-free. However, this order itself has no influence on the probability. In fact for two $p$-causes $\Pi, \Phi$ with $\Pi \preceq \Phi$ it can happen that for $\pi \in \Pi, \phi\in \Phi $ we have $\Pr(\goal \mid \pi) \geq \Pr(\goal \mid \phi)$.
	This partial order admits a minimal element which is a regular language over $S$ and which plays a crucial role for finding optimal causes in \Cref{sec:costs}.

\begin{restatable}[Canonical $p$-cause]{proposition}{canonicalpath} \label{prop: canonical path}
  Let
	\begin{gather*}
		\Theta = \left \{s_0 \cdots s_n \in \Paths_\fin(M) \;\middle| \; s_n\in S_p\text{ and for all $i<n$: } \:s_i\notin S_p\right \}.
	\end{gather*}
  Then $\Theta$ is a regular $p$-cause (henceforth called the \emph{canonical} $p$-cause) and for all $p$-causes $\Pi $ we have $\Theta \preceq \Pi$.
\end{restatable}

We now introduce an MDP associated with $M$ whose schedulers correspond to the $p$-causes of $M$.
This is useful both to represent $p$-causes and for algorithmic questions we consider later.
\begin{definition}[$p$-causal MDP]\label{rem: MDP construction}
  For the DTMC $M = (S,s_0,\mathbf{P})$ define the \emph{$p$-causal MDP} $\mdp_p(M) = (S,\{continue,pick\},s_0,\mathbf{P}')$ associated with $M$, where $\mathbf{P}'$ is defined as follows:
  \begin{align*}
    \mathbf{P'}(s,continue,s') &= \mathbf{P}(s,s') \text{ for all } s,s' \in S \\
    \mathbf{P'}(s,pick,\goalst) &= 
    \begin{cases}
      1 & \text{ if } s \in S_p \\
      0 & \text{ otherwise}
    \end{cases}
  \end{align*}
  Given a weight function $c$ on $M$, we consider $c$ also as weight function on $\mdp_p(M)$.
\end{definition}

\begin{example}
	\Cref{fig: MDP construction accumulated} demonstrates the $p$-causal MDP construction of $\mdp_p(M)$. The black edges are transitions of $M$, probabilities are omitted. Let us assume $S_p \backslash \{\goalst\}= \{s_1,s_3,s_4 \}$. To construct $\mdp_p(M)$ one adds transitions for the action $pick$, as shown by red edges.
	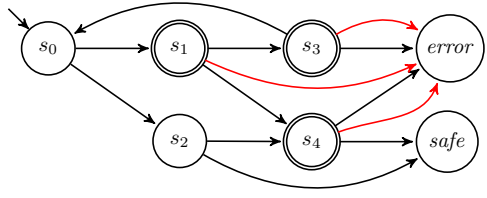
\begin{figure}[tbp]
		\centering
		\begin{tikzpicture}[->,>=stealth',shorten >=1pt,auto,node distance=0.5cm, semithick]
			
			\node[scale=0.8, state] (s0) {$s_0$};
			\node[scale=0.8, state,accepting] (s1) [right = 1 of s0]{$s_1$};
			\node[scale=0.8, state] (s2) [below = 0.5 of s1]{$s_2$};
			\node[scale=0.8, state,accepting] (s3) [right = 1 of s1]{$s_3$};
			\node[scale=0.8, state,accepting] (s4) [below = 0.5 of s3]{$s_4$};
			\node[scale=0.8, state] (f) [right = 1 of s4]{$\failst$};
			\node[scale=0.8, state] (g) [right = 1 of s3]{$\goalst$};
			
			\draw[<-] (s0) -- ++(-0.55,0.55);
			\draw (s0) -- (s1) node[pos=0.5,scale=0.8] {};
			\draw (s0) -- (s2) node[pos=0.5,scale=0.8] {};
			\draw (s1) -- (s3) node[pos=0.5,scale=0.8] {};
			\draw (s1) -- (s4) node[pos=0.5,scale=0.8] {};
			\draw (s2) -- (s4) node[pos=0.5,scale=0.8] {};
			\draw (s2) to [out=330, in=210] (f) node[pos=0.5,scale=0.8] {};
			\draw (s3) to [out=150,in=30] (s0) node[pos=0.5,scale=0.8] {};
			\draw (s3) -- (g) node[pos=0.5,scale=0.8] {};
			\draw (s4) -- (g) node[pos=0.5,scale=0.8] {};
			\draw (s4) -- (f) node[pos=0.5,scale=0.8] {};
			\draw[red] (s1) to [out=335, in=205] (g) node[pos=0.5,scale=0.8] {};
			\draw[red] (s3) to [out=30,in=150] (g) node[pos=0.5,scale=0.8] {};
			\draw[red] (s4) to [out=20,in=250] (g) node[pos=0.5,scale=0.8] {};
			
		\end{tikzpicture}
		\caption{Illustration of the $p$-causal MDP construction} \label{fig: MDP construction accumulated}
	\end{figure}
\end{example}

Technically, schedulers are defined on all finite paths of an MDP $\mathcal{M}$. 
However, under any scheduler, there are usually paths that cannot be obtained under the scheduler.
Thus we define an equivalence relation $\equiv$ on the set of schedulers of $\mathcal{M}$ by setting $\mathfrak{S} \equiv \mathfrak{S}'$ if $\Paths(\mathcal{M}_{\mathfrak{S}})=\Paths(\mathcal{M}_{\mathfrak{S}'})$.
Note that two schedulers equivalent under $\equiv$ behave identically.

\begin{lemma} \label{lem: scheduler trafo}
	There is a one-to-one correspondence between equivalence classes of schedulers in $\mdp_p(M)$ w.r.t. $\equiv$ and $p$-causes in $M$ for $\goal$.
\end{lemma}
\begin{proof}
	Given a $p$-cause $\Pi$ for $\goal$ in $M$, we construct the equivalence class of scheduler $[\mathfrak{S}_{\Pi}]$ by defining $\mathfrak{S}_{\Pi}(\hat\pi)=pick$ if $\hat\pi \in \Pi$, and otherwise $\mathfrak{S}_{\Pi}(\hat\pi)=continue$. Vice versa, given an equivalence class $[\mathfrak{S}]$ of schedulers, we define the $p$-cause 
	\begin{align*}
		\Pi_{\mathfrak{S}}= \left \{\hat\pi \in \Paths_\fin(M) \; \middle| \; 
		\begin{aligned}
			& \mathfrak{S}(\hat\pi)=pick \text{ or $\hat\pi$ ends in $\goalst$ and} \\
			& \text{$\mathfrak{S}$ does not choose $pick$ on any prefix of $\hat\pi$}
		\end{aligned} 
		\right \}
	\end{align*}
	Since $pick$ can only be chosen once on every path in $\Paths(\mathcal{M}_{\mathfrak{S}})$, it is easy to see that $\mathfrak{S} \equiv \mathfrak{S}'$ implies $\Pi_{\mathfrak{S}} = \Pi_{\mathfrak{S}'}$.
	Note that every $\hat\pi \in \Pi_\mathfrak{S}$ is a $p$-critical prefix since it ends in $S_p$ and every path in $\goal$ is covered since either $pick$ is chosen or $\hat\pi$ ends in $\goalst$.
	Furthermore, the second condition makes $\Pi$ prefix-free. 
	\qed
\end{proof}

\subsection{Types of $p$-causes and induced monitors}
We now introduce two classes of $p$-causes which have a comparatively simple representation, and we explain what classes of schedulers they correspond to in the $p$-causal MDP and how monitors can be derived for them.

\begin{definition}[State-based $p$-cause]\label{def: state-based}
	A $p$-cause $\Pi$ is \emph{state-based} if there exists a set of states $Q \subseteq S_p$ such that $\Pi = \{s_0 \ldots s_n\in \Paths_\fin(M) \mid s_n \in Q \text{ and } \forall i<n: \; s_i \notin Q \}$.
\end{definition}
State-based $p$-causes correspond to memoryless schedulers of $\mdp_p(M)$ which choose $pick$ exactly for paths ending in $Q$.
For DTMCs equipped with a weight function we introduce \emph{threshold-based $p$-causes}:
\begin{definition}[Threshold-based $p$-cause]
  A $p$-cause $\Pi$ is \emph{threshold-based} if there exists a map $T: S_p \to \mathbb{Q}\cup\{\infty\}$ such that
  \[\Pi = \left\{ s_0 \cdots s_n\in \Paths_\fin(M) \; \middle| \; 
  \begin{aligned}&s_0 \cdots s_n \in \operatorname{pick}(T) \text{ and}\\
  &s_0 \cdots s_i \notin \operatorname{pick}(T) \text{ for } i < n\end{aligned}\right\}\]
  where $\operatorname{pick}(T) = \{s_0 \ldots s_n \in \Paths_\fin(M) \mid  s_n\in S_p \text{ and } c(s_0\ldots s_n) < T(s_n)\}$.
\end{definition}
Threshold-based $p$-causes correspond to a simple class of \emph{weight-based} schedulers of the $p$-causal MDP, which base their decision in a state only on whether the current weight exceeds the threshold or not.
Intuitively, threshold-based $p$-causes are useful if triggering an alarm causes costs while reaching a safe state does not (see \Cref{sub:pexpcost}): The idea is that cheap paths (satisfying $c(s_0 \ldots s_n) < T(s_n)$) are picked for the $p$-cause, while expensive paths are continued in order to realize the chance (with probability $\leq1{-}p$) that a safe state is reached and therefore the high cost that has already been accumulated is avoided. 

The concept of $p$-causes can be used as a basis for monitors that raise an alarm as soon as a state sequence in the $p$-cause has been observed. State-based $p$-causes have the advantage that they are realizable by ``memoryless'' monitors that only need the information on the current state of the Markov chain. 
Threshold-based monitors additonally need to track the weight that has been accumulated so far until the threshold value of the current state is exceeded.
So, the memory requirements of monitors realizing a threshold-based $p$-cause are given by the logarithmic length of the largest threshold value for $S_p$-states.
All algorithms proposed in~\Cref{sec:costs} for computing cost-minimal $p$-causes will return $p$-causes that are either state-based or threshold-based with polynomially bounded memory requirements.

\subsection{Comparison to prima facie causes}

The work \cite{KleinbergM2009} presents the notion of \emph{prima facie causes} in DTMCs where both causes and events are formalized as PCTL state formulae.
In our setting we can equivalently consider a state $\goalst \in S$ as the effect and a state subset $C \subseteq S$ constituting the cause.
We then reformulate \cite[Definition 4.1]{KleinbergM2009} to our setting.

\begin{definition}[cf. \cite{KleinbergM2009}] \label{def: Kleinberg Causality}
	A set $C \subseteq S$ is a \emph{$p$-prima facie cause} of $\goal$ if the following three conditions hold:
	\begin{enumerate}
		\item The set $C$ is reachable from the initial state and $\goalst \notin C$.
		\item $\forall s \in C: \;\Pr_{s}(\goal) \geq p$
		\item $\Pr_{s_0}(\goal)<p$
	\end{enumerate}
\end{definition}

The condition $p>\Pr_{s_0}(\goal)$ we discussed for $p$-causes is hard-coded here as (3).
In \cite{KleinbergM2009} the value $p$ is implicitly existentially quantified and thus conditions (2) and (3) can be combined to $\Pr_s(\goal) > \Pr_{s_0}(\goal)$ for all $s \in C$.
This encapsulates the probability-raising property.
However, $\goalst$ may be reached while avoiding the cause $C$, so $p$-prima facie causes do not entail the \emph{counterfactuality} principle.
\Cref{def:p-cause} can be seen as an extension of $p$-prima facie causes by virtue of the following lemma:

\begin{lemma}
	For $p>\Pr_{s_0}(\goal)$ every $p$-prima facie cause induces a state-based $p$-cause.
\end{lemma}
\begin{proof}
	Let $C \subseteq S$ be a $p$-prima facie cause.
	By condition (1) and (2) of \Cref{def: Kleinberg Causality} we have $C \subseteq S_p \backslash \{\goalst\}$. 
	Since every path reaching $\goalst$ trivially visits a state in $ 
	Q : = C \cup \{\goalst\}\subseteq S_p$, the set $\Pi = \{s_0 \ldots s_n\in \Paths_\fin(M) \mid s_n \in Q \text{ and } \forall i<n: \; s_i \notin Q\}$ is a state-based $p$-cause.
	\qed
\end{proof}

\section{Costs of $p$-causes}
\label{sec:costs}

 In this section we fix a DTMC $M$ with state space $S$, unique initial state $s_0$, unique target and safe state $\goalst,\failst \in S$ and a threshold $p \in (0,1]$.
 As motivated in the introduction, we equip the DTMC of our model with a \emph{weight function} $c\colon S \to \mathbb{Q}$ on states and consider the induced accumulated weight function $c\colon \Paths_\fin(M)\to\mathbb{Q}$. These weights typically represent resources spent, e.g., energy, time, material, etc.

\subsection{Expected cost of a $p$-cause}

\begin{definition}[Expected cost]
 	Given a $p$-cause $\Pi$ for $\goal$ in $M$ consider the random variable $\mathcal{X}: \Paths(M) \to \mathbb{Q}$ with
	\[ \mathcal{X}(\pi)=c(\hat\pi) \text{ for } 
	\begin{cases}
	\hat\pi \in \Pi \cap \Pref(\pi) & \text{ if such } \hat\pi \text{ exists} \\
	\hat\pi \in \Pref(\pi) \text{ minimal with } \last(\hat\pi)=\failst & \text{ otherwise}.
	\end{cases} \]
	Since $\Pr_{s_0}(\lozenge \{\goalst, \failst\})=1$, paths not falling under the two cases above have measure $0$. Then the \emph{expected cost} $\expcost(\Pi)$ of $\Pi$ is the expected value of $\mathcal{X}$.
\end{definition}

The expected cost is a means by which the efficiency of causes for monitoring purposes can be estimated. Assume a $p$-cause $\Pi$ is used to monitor critical scenarios of a probabilistic system. This means that at some point either a critical scenario is predicted by the monitor (i.e., the execution seen so far lies in $\Pi$), or the monitor reports that no critical scenario will arise (i.e., $\failst$ has been reached) and can therefore be turned off.
If the weight function on the state space is chosen such that it models the cost of monitoring the respective states, then $\expcost(\Pi)$ estimates the average total resource consumption of the monitor. 

We say that a $p$-cause $\Pi$ is $\expcost$-\emph{minimal} if for all $p$-causes $\Phi$ we have $\expcost(\Pi) \leq \expcost(\Phi)$. By $\expcost^{\min}$, we denote the value $\expcost(\Pi)$ of any $\expcost$-minimal $p$-cause $\Pi$.

\begin{restatable}{theorem}{expcostComplexity} \label{thm: complexity expcost}
	\begin{enumerate}
	\item Given a non-negative weight function $c:S \to \mathbb{Q}_{\geq 0}$, the canonical $p$-cause $\Theta$ from \Cref{prop: canonical path} is $\expcost$-minimal.

	\item For an arbitrary weight function $c:S \to \mathbb{Q}$, an $\expcost$-minimal and state-based $p$-cause $\Pi$ and $\expcost^{\min}$ can be computed in polynomial time.
	\end{enumerate}
\end{restatable}
\begin{proofsketch}
	The statement (1) follows from the fact that if $\Pi \preceq \Phi$ holds for two $p$-causes, then we have $\expcost(\Pi) \leq \expcost(\Phi)$, which is shown in the appendix.
	The value $\expcost^{\min}=\expcost(\Theta)$ can then be computed in polynomial time using methods for expected rewards in Markov chains\cite[Section 10.5]{BaierK2008}.
	
  To show (2), we reduce our problem to the \emph{stochastic shortest path problem} (SSP)\cite{BT91} from $s_0$ to $\{\goalst, \failst\}$.
	By \Cref{lem: scheduler trafo} equivalence classes of schedulers in $\mdp_p(M)$ are in one-to-one correspondence with $p$-causes in $M$.
  Let $\Pi_\mathfrak{S}$ be a $p$-cause associated with a representative scheduler $\mathfrak{S}$.
  One can show that $\expcost(\Pi_\mathfrak{S})$ is equal to the expected accumulated weight of paths under scheduler $\mathfrak{S}$ in $\mdp_p(M)$ upon reaching $\{\goalst,\failst\}$.
  A scheduler $\mathfrak{S}^*$ minimizing this value can be computed in polynomial time by solving the SSP in $\mdp_p(M)$ \cite{BT91}, and the algorithm returns a memoryless such $\mathfrak{S}^*$.
   It follows that $\Pi_{\mathfrak{S}^*}$ is an $\expcost$-minimal and state-based $p$-cause.
	\qed
\end{proofsketch}

\subsection{Partial expected cost of a $p$-cause}
\label{sub:pexpcost}

In this section we study a variant of the expected cost where paths with no prefix in the $p$-cause are attributed zero costs.
A use case for this cost mechanism arises if the costs are not incurred by monitoring the system, but by the countermeasures taken upon triggering the alarm.
For example, an alarm might be followed by a downtime of the system, and the cost of this may depend on the current state and history of the execution.
In such cases there are no costs incurred if no alarm is triggered.

\begin{definition}[Partial expected cost]
	For a $p$-cause $\Pi$ for $\goal$ in $M$ consider the random variable $\mathcal{X}: \Paths(M) \to \mathbb{Q}$ with
	\begin{align*} \label{eq: expcost0 path variable}
	\mathcal{X}(\pi)= 
	\begin{cases}
	c(\hat\pi)& \text{for } \hat\pi \in \Pi \cap \Pref(\pi) \text{ if such } \hat\pi \text{ exists} \\
	0 & \text{otherwise}.
	\end{cases} 
	\end{align*}
	The \emph{partial expected cost} $\pexpcost(\Pi)$ of $\Pi$ is the expected value of $\mathcal{X}$.
\end{definition}
The analogous statement to \Cref{thm: complexity expcost} (1) does not hold for partial expected costs, as the following example shows.
\begin{example}
	\label{ex: counterexample monotonicity}
	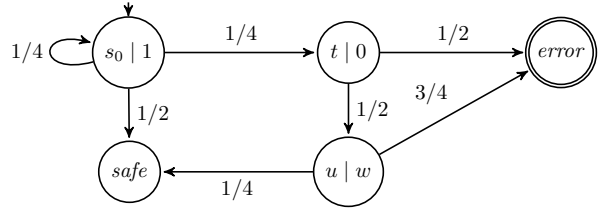
\begin{figure}[t]
			\centering
			\begin{tikzpicture}[->,>=stealth',shorten >=1pt,auto,node distance=0.5cm, semithick]
				\node[scale=0.8, state] (s0) {$s_0\mid 1$};
				\node[scale=0.8, state] (t) [right=2 of s0]{$t\mid 0$};
				\node[scale=0.8, state] (u) [below = 0.7 of t] {$u\mid w$};
				\node[scale=0.8, state, accepting] (g) [right =1.92 of t] {$\goalst$};
				\node[scale=0.8, state] (f) [left =2 of u] {$\failst$};
				
				\draw[<-] (s0) --++(0,0.7);
				\path (s0) edge[loop left] node[scale=0.8] {$1/4$} (s0);
				\draw (s0) -- (t) node[pos=0.5,scale=0.8] {$1/4$};
				\draw (s0) -- (f) node[pos=0.5,scale=0.8] {$1/2$};
				\draw (t) -- (u) node[pos=0.5,scale=0.8] {$1/2$};
				\draw (t) -- (g) node[pos=0.5,scale=0.8] {$1/2$};
				\draw (u) -- (g) node[pos=0.5,scale=0.8] {$3/4$};
				\draw (u) -- (f) node[pos=0.5,scale=0.8] {$1/4$};	
			\end{tikzpicture}
		\caption{An example showing that the partial expected cost is not monotonous on $p$-causes when $c$ is non-negative.} \label{fig: counterexample monotonicity}
	\end{figure}
	Consider the Markov chain depicted in Figure \ref{fig: counterexample monotonicity}.
	For $p=1/2$ and $\goal$ we have $S_p = \{t, u, \goalst\}$.
	The canonical $p$-cause is $\Theta= \{s_0^k t\mid k \geq 1 \}$ with  $\pexpcost(\Theta)= \sum_{k \geq 1} (1/4)^k \cdot k  = 4/9$.
    Now let $\Pi$ be any $p$-cause for $\goal$. 
	If the path $s_0^\ell t$ belongs to $\Pi$, then it contributes $(1/4)^\ell \cdot \ell$ to $\pexpcost(\Pi)$.
	If instead the paths $s_0^\ell t  \goalst$ and $s_0^\ell t  u \goalst$ belong to $\Pi$, they contribute $(1/4)^\ell \cdot 1/2 \cdot \ell +(1/4)^\ell \cdot 1/2 \cdot 3/4 \cdot (\ell+w)$.
	So, the latter case provides a smaller $\pexpcost$ if $l > 3 w$, and the $\pexpcost$-minimal  $p$-cause is therefore
	\[ \Pi= \{s_0^k t \mid 1 \leq k \leq 3 w \} \cup \{s_0^k t \goalst, s_0^k t u \mid 3 w < k \} . \]
	For $w=1$, the expected cost of this $p$-cause is $511/1152=4/9-1/1152$. 
	So, it is indeed smaller than $\pexpcost(\Theta)$.
\end{example}

\begin{restatable}{theorem}{pexpcostPseudoP}
	\label{thm: expcost0 minimal path}
	Given a non-negative weight function $c \colon S \to \mathbb{Q}_{\geq 0}$, a $\pexpcost$-minimal and threshold-based $p$-cause $\Pi$, and the value $\pexpcost^{\min}$, can be computed in pseudo-polynomial time. $\Pi$ has a polynomially bounded representation.
\end{restatable}

\begin{proofsketch}
	For the pseudo-polynomial time bound we apply the techniques from \cite{fossacs2019} to optimize the partial expected cost in MDPs to the $p$-causal MDP $\mdp_p(M)$.
  It is shown in \cite{fossacs2019} that there is an optimal scheduler whose decision depends only on the current state and accumulated weight
   and that such a scheduler and its partial expectation can be computed in pseudo-polynomial time. 
   It is further shown that  a rational number $K$ can be computed in polynomial time such that for accumulated weights above $K$, an optimal scheduler has to  minimize the probability to reach $\goalst$. In our case, this means choosing the action $continue$.
  Due to the special structure of   $\mdp_p(M)$, we can further show that there is indeed a threshold $T(s)$ for each state $s$ such that action $pick$ is optimal after a path $\hat\pi$ ending in $s$ if and only if $c(\hat\pi)<T(s)$. So, a threshold-based $\pexpcost$-minimal $p$-cause can be computed in pseudo-polynomial time. Furthermore, we have $T(s)<K$ for each state $s$ and as $K$ has a polynomially bounded representation the same applies to the values $T(s)$ for all states $s$. 
\qed
\end{proofsketch}

Since the causal MDP $\mdp_p(M)$ has a comparatively simple form, one could expect that one can do better than the pseudo-polynomial algorithm obtained by reduction to \cite{fossacs2019}. Nevertheless, in the remainder of this section we argue that computing a $\pexpcost$-minimal $p$-cause is computationally hard, in contrast to $\expcost$ (cf. \Cref{thm: complexity expcost}). For this we recall that the complexity class $\pp$ \cite{Gill1977} is characterized as the class of languages $\lang$ that have a probabilistic polynomial-time bounded Turing machine $M_{\lang}$ such that for all words $\tau$ one has $\tau \in \lang$ if and only if $M_{\lang}$ accepts $\tau$ with probability at least $1/2$ (cf. \cite{HaaseK15}).
We will use polynomial Turing reductions, which, in contrast to many-one reductions, allow querying an oracle that solves the problem we reduce to a polynomial number of times.
A polynomial time algorithm for a problem that is $\pp$-hard under polynomial Turing reductions would imply that the polynomial hierarchy collapses~\cite{Toda1991}.
We reduce the $\pp$-complete cost-problem stated in \cite[Theorem 3]{HaaseK14} to the problem of computing $\pexpcost^{\min}$.

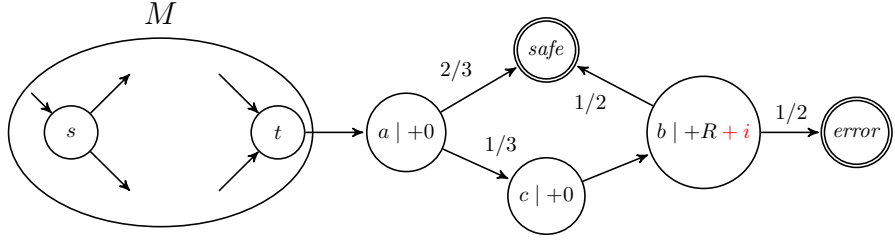
\begin{figure}[tbp]
	\centering
	\begin{tikzpicture}[->,>=stealth',shorten >=1pt,auto,node distance=0.5cm, semithick]
	
	\node[scale=0.8, state] (s) {$s$};
	\node[scale=0.8, state] (t) [right = 2 of s]{$t$};
	\node[scale=0.8, state] (ai) [right = 0.8 of t]{$a \mid +0$};
	\node[scale=0.8, state] (ci) [right = 0.8 of ai, yshift=-30pt]{$c \mid +0$};
	\node[scale=0.8, state, accepting] (fail) [above = 1 of ci]{$\failst$};
	\node[scale=0.8, state] (bi) [right = 0.8 of ci,yshift=30pt]{$b \mid +R \color{red} +i$};
	\node[scale=0.8, state, accepting] (goal) [right = 0.8 of bi] {$\goalst$};
	\node[res, scale=5] (M) [left = 0.65 of ai]{};
	\node[scale=1] (M text) [above = 0.01 of M]{\large $M$};

	\draw[<-] (s) -- ++(-0.55,0.55);
	\draw[->] (s) -- ++(0.8,0.8);
	\draw[->] (s) -- ++(0.8,-0.8);
	\draw[<-] (t) -- ++(-0.8,0.8);
	\draw[<-] (t) -- ++(-0.8,-0.8);
	\draw (t) -- (ai) node[pos=0.5,scale=0.8] {};
	\draw (ai) -- (fail) node[pos=0.5,scale=0.8] {$2/3$};
	\draw (ai) -- (ci) node[pos=0.5,scale=0.8] {$1/3$};
	\draw (ci) -- (bi) node[pos=0.5,scale=0.8] {};
	\draw (bi) -- (fail) node[pos=0.5,scale=0.8] {$1/2$};
	\draw (bi) -- (goal) node[pos=0.5,scale=0.8] {$1/2$};
	\end{tikzpicture}
	\caption{The DTMCs $N_i$ for $i=0,1$} \label{fig:sketch PP reduction}
\end{figure}

\begin{restatable}{theorem}{pexpcostPP}\label{thm: expcost0 PP-hard}
	Given an acyclic DTMC $M$, a weight function $c \colon S \to \mathbb{N}$ and a rational $\vartheta \in \mathbb{Q}$, deciding whether $\pexpcost^{\min}\leq \vartheta$ is $\pp$-hard under Turing reductions.
\end{restatable}

\begin{proofsketch}
We sketch a Turing reduction from the following problem which is shown to be $\pp$-hard in\cite{HaaseK14}:
Given an acyclic DTMC $M$ over state space $S$ with initial state $s$, absorbing state $t$ such that $\Pr_{s}(\lozenge t)=1$, weight function $c:S \to \mathbb{N}$ and natural number $R \in \mathbb{N}$, decide whether 
		\[\Pr_M(\{\pi \in \Paths(M) \mid c(\pi) \leq R\}) \geq 1/2.\]

In an acyclic Markov chain $M$ the values of $\pexpcost$ have a polynomially bounded binary representation as shown in the appendix.
This allows for a binary search to compute $\pexpcost^{\min}$ with polynomially many calls to the corresponding threshold problem. We use this procedure in a polynomial-time Turing reduction.

Let now $M$ be a Markov chain as in \cite[Theorem 3]{HaaseK14} and let $R$ be a natural number. We construct two Markov chains $N_0$ and $N_1$ depicted in Figure \ref{fig:sketch PP reduction}.
The $\pexpcost$-minimal $p$-cause in both Markov chains consists of all paths reaching $c$ with weight $\leq R$ and all paths reaching $\goalst$ that do not have a prefix reaching $c$ with weight $\leq R$.
The difference between the values $\pexpcost^{\min}$ in the two Markov chains depends only on the probability of paths in the minimal $p$-cause collecting the additional weight $+1$ in $N_1$.
This probability is $\frac{1}{6}\Pr_M(\{\pi \in \Paths(M) \mid c(\pi) \leq R\})$.
By repeatedly using the threshold problem to compute $\pexpcost^{\min}$ in $N_0$ and $N_1$ as described above, we can hence decide the problem from \cite[Theorem 3]{HaaseK14}. More details can be found in the appendix. \qed
\end{proofsketch}

\subsection{Maximal cost of a $p$-cause}

In practice, the weight function on the Markov chain potentially models resources for which the available consumption has a \emph{tight} upper bound. For example, the amount of energy a drone can consume from its battery is naturally limited. Instead of just knowing that \emph{on average} the consumption will lie below a given bound, it is therefore often desirable to find monitors whose costs are guaranteed to lie below this limit for (almost) any evolution of the system. 

\begin{definition}[Maximal cost]
	Let $\Pi$ be a $p$-cause for $\goal$ in $M$. We define the \emph{maximal cost} of $\Pi$ to be
	\begin{align*}
		\maxcost(\Pi) &= \sup \{c(\hat\pi) \mid \hat\pi \in \Pi \}.
	\end{align*}
\end{definition}

The maximal cost of a $p$-cause is a measure for the worst-case resource consumption among executions of the system. 
Therefore, by knowing the minimal value $\maxcost^{\min}$ for $p$-causes one can ensure that there will be no critical scenario arising from resource management.
\begin{restatable}{theorem}{maxcostSummary} \label{thm: maxcost}
	\begin{enumerate}
		\item Given a non-negative weight function $c:S \to \mathbb{Q}_{\geq 0}$, the canonical $p$-cause $\Theta$ is $\maxcost$-minimal and $\maxcost^{\min}$ can be computed in time polynomial in the size of $M$.
		\item For an arbitrary weight function $c:S \to \mathbb{Q}$ a $\maxcost$-minimal and state-based $p$-cause $\Pi$ and $\maxcost^{\min}$ can be computed in pseudo-polynomial time.
		\item Given a rational $\vartheta \in \mathbb{Q}$, deciding whether $\maxcost^{\min} \leq \vartheta$ is in $\np \cap \conp$. 
	\end{enumerate}
\end{restatable}
\begin{proofsketch}
	To show (1) it suffices to note that for non-negative weight functions $\maxcost$ is monotonous with respect to the partial order $\preceq$ on $p$-causes.
	Therefore $\Theta$ is $\maxcost$-minimal.
	For (2) we reduce the problem to a max-cost reachability game as defined in \cite{BGHM2015}.
	The algorithm from \cite{BGHM2015} computes the lowest maximal cost and has a pseudo-polynomial time bound.
	By virtue of the fact that the minimizing player has a memoryless strategy we can compute a set of states $Q \subseteq S_p$ on which a $\maxcost$-minimal $p$-cause $\Pi$ is based upon.
	In order to show (3) we reduce the max-cost reachability game from (2) further to mean-payoff games, as seen in \cite{Chatterjee2017}.
	Mean-payoff games are known to lie in $\np \cap \conp$ \cite{Chatterjee2017}.
	\qed
\end{proofsketch}

\subsection{Instantaneous cost} 

The given weight function $c$ on states can also induce an instantaneous weight function $c_\inst:\Paths_\fin(M) \to \mathbb{Q}$ which just takes the weight of the state visited last, i.e., $c_\inst(s_0 \cdots s_n)=c(s_n)$.
This yields an alternative cost mechanism intended to model the situation where the cost of repairing or rebooting only depends on the current state, e.g., the altitude an automated drone has reached.

We add the subscript `$\inst$' to the three cost variants, where the accumulative weight function $c$ has been replaced with the instantaneous weight function $c_\inst$, the error state is replaced by an error set $E$ and the safe state is replaced by a set of terminal safe states $F$.
Thus we optimize $p$-causes for $\lozenge E$ in $M$.

\begin{restatable}{theorem}{instcost} \label{thm: instcost}
	For $\expcost_\inst$, $\pexpcost_\inst$, and $\maxcost_\inst$ a cost-minimal $p$-cause $\Pi$ and the value of the minimal cost can be computed in time polynomial in $M$.
	In all cases $\Pi$ can be chosen to be a state-based $p$-cause.
\end{restatable}
\begin{proofsketch}
	We first note that $\pexpcost_\inst$ can be reduced to $\expcost_\inst$ by setting the weight of all states in $F$ to $0$.
	We then construct an MDP (different from $\mdp_p(M)$) which emulates the instantaneous weight function using an accumulating weight function.
	Thus, finding an $\expcost_\inst$-minimal $p$-cause $\Pi$ reduces to the SSP from \cite{BT91}, which can be solved in polynomial time.
	The solution admits a memoryless scheduler and thus $\Pi$ is state-based in this case.
	
	For $\maxcost_\inst$ we order the states in $S_p$ by their cost and then start iteratively removing the states with lowest cost until $E$ is not reachable anymore.
	The set $Q$ of states which where removed induce a state-based $\maxcost_\inst$-minimal $p$-cause $\Pi$.
	This gives us a polynomial time procedure to compute $\maxcost_\inst^{\min}$ and $\Pi$.
	\qed
\end{proofsketch}

\section{Conclusion}
We combined  the counterfactuality principle and the probability-raising property into the notion of $p$-causes in DTMCs. 
In order to find suitable $p$-causes we defined different cost models and gave algorithms to compute corresponding cost-minimal causes.

Cyber-physical systems are often not fully probabilistic, but involve a certain amount of control in form of decisions depending on the system state.
Such systems can be modeled by MDPs, to which we intend to generalize the causality framework presented here.
Our approach also assumes that the probabilistic system described by the Markov chain is fully observable. By observing execution traces instead of paths of the system, generalizing the notion of $p$-causes to hidden Markov models is straightforward. However, the corresponding computational problems exhibit additional difficulties which we address in future work.

\bibliographystyle{splncs04}
\bibliography{lit2}

\section{Appendix}

We will denote nondeterministic finite automata as tuples $\A=(Q, \Sigma,Q_0,T,F)$, where $Q$ is the set of states, $\Sigma$ is the alphabet, $Q_0\subseteq Q$ is the initial state, $T\colon Q\times\Sigma\to 2^Q$ is the transition function, and $F\subseteq Q$ are the final states. A transition $q'\in T(q,\alpha)$ is also denoted by $q\overset{\alpha}{\mapsto} q'$.

\canonicalpath*
\begin{proof}
	The first condition $s_n \in S_p = \{s \in S \mid \Pr_s(\goal) \geq p\}$ ensures that every path in $\Theta$ is a $p$-critical prefix for $\goal$ in $M$.
	The second condition ensures that no proper prefix is a $p$-critical prefix, and therefore $\Theta$ is prefix-free. Clearly, every path reaching $\goalst$ passes through $S_p$ since $\goalst$ is itself an element in $S_p$. These properties together show that $\Theta$ is a $p$-cause.
	
	Let $\Pi$ be another $p$-cause for $\goal$ in $M$ and consider an arbitrary path $s_0 
	\ldots s_n \in \Pi$.
	This means $s_n \in S_p$.
	If we have for all $i<n$ that $\Pr_{s_i}(\goal)<p$, then $s_0 \ldots s_n \in \Theta$ and there is nothing to prove.
	Therefore assume that there exists $i<n$ such that $s_i \in S_p$.
	For minimal such $i$ we have $s_0 \ldots  s_i \in \Theta$ and thus $\Theta \preceq \Pi$.
	
	To prove that $\Theta$ is regular, consider the deterministic finite automaton $\A=(Q, \Sigma,q_0,T,F)$ for $Q=S$, $\Sigma=S\cup\{\iota\}$, $q_0=\iota$, and final states $F=S_p$. The transition relation is given by
	\[T= \{\iota \iotato s_0\}\cup\{s \sto s' \mid P(s,s')>0 \wedge s \notin S_p \}\] 
	i.e., all states in $S_p$ are terminal.
	We argue that this automaton describes the language $\Theta$:

	If $s_0 \in S_p$ then there are no transitions possible and $\lang_\A=\{s_0\}$ and $\Theta = \{s_0\}$ by definition.
	
	Let $\lang_\A \subseteq \Theta$ and let $s_0 \ldots s_n$ be a word accepted by $\A$.
	By definition we have $s_n \in S_p$
	For any state $s \in S_p$ visited by $s_0 \cdots s_n$ we know that $s=s_n$ since $S_p$ is terminal.
	This means $\forall i<n. \; \Pr_{s_i}(\goal) < p$ and thus $s_0 \ldots s_n \in \Theta$.
	
	Now let $\lang_\A \supseteq \Theta$ and $s_0 \ldots s_n \in \Theta$.
	A run for $s_0 \ldots s_n$ in $\A$ is possible, since we have an edge for every transition in $M$ except for outgoing edges of $S_p$.
	These exceptions are avoided, since for $i<n$ we have $\Pr_{s_i}(\goal) < p$ and thus no state but the last one is in $S_p$.
	On the other hand $s_n$ is in $S_p$ by assumption and thus the word $s_0 \ldots s_n$ is accepted by $\A$.
	\qed
\end{proof}

In the following proof, let $\Cyl(\Pi) = \bigcup_{\pi\in\Pi} \Cyl(\pi)$ for $\Pi\subseteq\Paths_\fin(M)$.

\expcostComplexity*
\begin{proof}
	To show (1), we claim that for two $p$-causes $\Pi$ and $\Phi$ with $\Pi\preceq \Phi$ we have $\expcost(\Pi) \leq \expcost(\Phi)$. 
	Let $\mathcal{X}_\Pi$ and $\mathcal{X}_\Phi$ be the corresponding random variables of $\expcost(\Pi)$ and $\expcost(\Phi)$, respectively.
	We prove $\expcost(\Pi) \leq \expcost(\Phi)$ by showing for almost all $\pi \in \Paths(M)$ that $\mathcal{X}_\Pi(\pi) \leq \mathcal{X}_\Phi(\pi)$. In the argument below we ignore cases that have measure $0$.
	
	From $\Pi\preceq \Phi$ it follows that $\Cyl(\Pi) \supseteq \Cyl(\Phi)$. We now deal separately with the three cases obtained by the partition
	\[ \Paths(M)= \Cyl(\Phi) \ \dot \cup \ \Cyl(\Pi) \backslash \Cyl(\Phi) \ \dot \cup \ \Paths(M) \backslash \Cyl(\Pi). \]
	
	For $s_0s_1 \ldots \in \Cyl(\Phi)$ both random variables consider the unique prefix of $s_0s_1 \ldots$ in the respective $p$-cause.
	For any $s_0 \ldots s_n \in \Phi$ we know by definition of $\preceq$ that there is $s_0 \ldots s_m \in \Pi$ with $m \leq n$.
	Therefore
	\[ \mathcal{X}_\Pi(s_0s_1 \ldots)=c(s_0 \ldots s_m) \leq c(s_0 \ldots s_n)= \mathcal{X}_\Phi(s_0s_1 \ldots). \]
	
	For $s_0s_1 \ldots \in \Cyl(\Pi) \backslash \Cyl(\Phi)$ the random variable $\mathcal{X}_\Pi$ takes the unique prefix $s_0\ldots s_m$ in $\Pi$, whereas $\mathcal{X}_\Phi$ takes the path $s_0 \ldots s_n$ such that $s_n = \failst$.
	Since $\Pr_{\failst}(\lozenge\goalst) = 0$, we have almost surely that $m<n$, and therefore
	\[ \mathcal{X}_\Pi(s_0s_1 \ldots)=c(s_0 \ldots s_m) \leq c(s_0 \ldots s_n)= \mathcal{X}_\Phi(s_0s_1 \ldots). \]
	
	For $s_0s_1 \ldots \in \Paths(M) \backslash \Cyl(\Pi)$ both random variables evaluate the same path and therefore $\mathcal{X}_\Pi(s_0s_1 \ldots)= \mathcal{X}_\Phi(s_0s_1 \ldots)$.
	Therefore we have for any $\pi \in \Paths(M)$ that $\mathcal{X}_\Pi(\pi) \leq \mathcal{X}_\Phi(\pi)$.
	
	To compute $\expcost^{\min}=\expcost(\Theta)$ consider the DTMC $M$ but change the transitions such that every state in $S_p$ is terminal.
	The value $\expcost(\Theta)$ is then the expected reward $\ExpRew(s_0 \models \lozenge S_p \cup \{\failst\})$ defined in \cite[Definition 10.71]{BaierK2008}. Expected rewards in DTMCs can be computed in polynomial time via a classical linear equation system \cite[Section 10.5.1]{BaierK2008}.
	
	We show (2) by reducing the problem of finding an $\expcost$-minimal $p$-cause to the stochastic shortest path problem (SSP) \cite{BT91} in $\mdp_p(M)$. For a scheduler $\mathfrak{S}$ of of $\mdp_p(M)$ denote the corresponding $p$-cause for $\goal$ in $M$ by $\Pi_{\mathfrak{S}}$ (cf. \Cref{lem: scheduler trafo}). In \cite{BT91} the authors define the value \emph{expected cost} of a scheduler $\mathfrak{S}$ in an MDP. This value with respect to the target set $\{\goalst, \failst\}$ coincides with $\expcost(\Pi_{\mathfrak{S}})$.
	
	The SSP asks to find a scheduler $\mathfrak{S}^*$ in $\mdp_p(M)$ minimizing the expected cost. It is shown in \cite{BT91} that a memoryless such $\mathfrak{S}^*$ and the value $\expcost^{\min}$ can be
	computed in polynomial time. The scheduler $\mathfrak{S}^*$  corresponds to an $\expcost$-minimal state-based $p$-cause $\Pi$.
	\qed
\end{proof}

\pexpcostPseudoP*
\begin{proof}
	Consider a scheduler $\mathfrak{S}$ of $\mdp_p(M)$ with weight function $c$ and recall that $\mdp_p(M)_\mathfrak{S}$ denotes Markov chain induced by $\mathfrak{S}$. Define 
	the random variable $\oplus_\mathfrak{S}  \goalst: \Paths(\mdp_p(M)_\mathfrak{S}) \to \mathbb{Q}$
	\[ \oplus_\mathfrak{S} \goalst(\pi)= \begin{cases}c(\pi)& \text{if } \pi \in \goal \\
	0 & \text{otherwise.}
	\end{cases} \]
	The \emph{partial expectation} ${\mathbb{PE}}^{\mathfrak{S}}$ of a scheduler $\mathfrak{S}$ in $\mdp_p(M)$ is defined as the expected value of $\oplus_\mathfrak{S} \goalst$.
	The minimal partial expectation is ${\mathbb{PE}}^{\min}= \inf_{\mathfrak{S}} {\mathbb{PE}}^{\mathfrak{S}}$.
	It is known that there is a scheduler obtaining the minimal partial expectation \cite{fossacs2019}.
	
	Then a $p$-cause $\Pi$ and the corresponding scheduler $\mathfrak{S}_{\Pi}$ satisfy $\pexpcost(\Pi)= \mathbb{PE}^{\mathfrak{S}_{\Pi}}$.
	A cost-minimal  scheduler for the partial expectation in an MDP with non-negative weights can be computed in pseudo-polynomial time by \cite{fossacs2019}.
	In this process we also compute $\pexpcost^{\min}={\mathbb{PE}}^{\min}$.
	
	Furthermore, we can show that once the action $continue$ is optimal in a state $s$ with accumulated weight $w$, it is also optimal for all weights $w'>w$:
	Suppose choosing $continue$ is optimal in some state $s$ if the accumulated weight is $w$. 
	Let $E$ be the partial expected accumulated weight that the optimal scheduler collects from then on and let $q = \Pr_{s}(\goal)$.
	The optimality of $continue$ implies that $E+w\cdot q \leq w$. For all $w^\prime>w$, this implies $E+w^\prime\cdot q \leq w^\prime$ as well.
	We conclude the existence of $T: S \to \mathbb{Q}$ such that $continue$ is optimal if and only if the accumulated weight is at least $T(s)$.
	If $pick$ is not enabled in a state $s$, we have $T(s) = 0$. 
	Therefore $\Pi$ is a threshold-based $p$-cause defined by $T$.
	As shown in \cite{fossacs2019}, there is a saturation point $K$ such that schedulers minimizing the partial expectation can be chosen to behave memoryless as soon as the accumulated weight exceeds $K$. This means that $T(s)$ can be chosen to be  either at most $K$ or $\infty$ for each state $s$.
	The saturation point $K$ and hence all thresholds $T(s)$ have a polynomial representation. 
	\qed
\end{proof}

\pexpcostPP*
\begin{proof}
		\begin{figure}[t]
		\centering
		\begin{tikzpicture}[->,>=stealth',shorten >=1pt,auto,node distance=0.5cm, semithick]
			
			\node[scale=0.8, state] (s) {$s$};
			\node[scale=0.8, state] (t) [right = 2 of s]{$t$};
			\node[scale=0.8, state] (ai) [right = 0.8 of t]{$a_i \mid +0$};
			\node[scale=0.8, state] (ci) [right = 0.8 of ai, yshift=-30pt]{$c_i \mid +0$};
			\node[scale=0.8, state, accepting] (fail) [above = 1 of ci]{$\failst$};
			\node[scale=0.8, state] (bi) [right = 0.8 of ci,yshift=30pt]{$b_i \mid +R \color{red} +i$};
			\node[scale=0.8, state, accepting] (goal) [right = 0.8 of bi] {$\goalst$};
			\node[res, scale=5] (M) [left = 0.65 of ai]{};
			\node[scale=1] (M text) [above = 0.01 of M]{\large $M$};

			\draw[<-] (s) -- ++(-0.55,0.55);
			\draw[->] (s) -- ++(0.8,0.8);
			\draw[->] (s) -- ++(0.8,-0.8);
			\draw[<-] (t) -- ++(-0.8,0.8);
			\draw[<-] (t) -- ++(-0.8,-0.8);
			\draw (t) -- (ai) node[pos=0.5,scale=0.8] {};
			\draw (ai) -- (fail) node[pos=0.5,scale=0.8] {$2/3$};
			\draw (ai) -- (ci) node[pos=0.5,scale=0.8] {$1/3$};
			\draw (ci) -- (bi) node[pos=0.5,scale=0.8] {};
			\draw (bi) -- (fail) node[pos=0.5,scale=0.8] {$1/2$};
			\draw (bi) -- (goal) node[pos=0.5,scale=0.8] {$1/2$};
		\end{tikzpicture}
		\caption{The DTMCs $N_i$ for $i=0,1$} \label{fig: PP reduction}
	\end{figure}
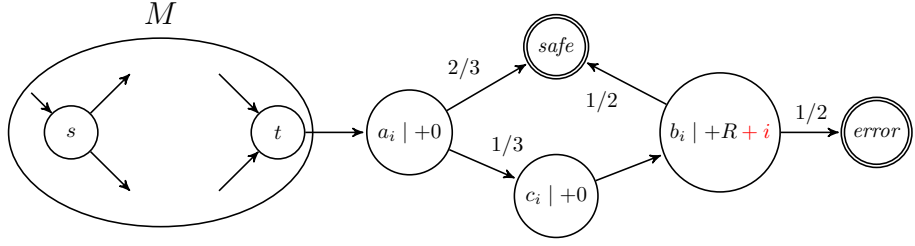
	We provide a Turing reduction from  the following problem that is shown to be $\pp$-hard in  \cite{HaaseK14}:
Given an acyclic DTMC $M$ over state space $S$ with initial state $s$, absorbing state $t$ such that $\Pr_{s}(\lozenge t)=1$, weight function $c:S \to \mathbb{N}$ and natural number $R \in \mathbb{N}$, decide whether 
		\[\Pr_M(\{\varphi \in \Paths(M) \mid c(\pi) \leq R\}) \geq \frac{1}{2}.\]

	Given such an acyclic DTMC $M$ we construct the  two DTMCs $N_0$ and $N_1$ depicted in Figure \ref{fig: PP reduction}.
	We consider the 
	 $\pexpcost$-minimal $p$-causes for $p=1/2$ in $N_i$ for $\goal$ and $i \in \{0,1\}$.
	 Suppose a path $\pi$ arrives at state $c_i$ with probability $\Pr_s(\pi)$ and accumulated weight $w$. We have to decide whether the path $\pi$ or the extension 
	 $\pi^\prime=\pi b_i \goalst$ should be included in the cost-minimal $p$-cause. The path $\pi^\prime$ has weight $w+R+i$ and probability $\Pr_s(\pi)/2$. We observe that $\pi$ is the optimal choice if 
	 \[
	 \Pr_s(\pi)\cdot w\leq \Pr_s(\pi)\cdot \frac{w+R+i}{2}.
	 \]
	 This is the case if and only if $w\leq R+i$. If $i=1$, and $w=R+i$, both choices are equally good and we decide to include the path $\pi$ in this case.
	 Hence, the $\pexpcost$-minimal $p$-causes for $p=1/2$ in $N_0$
	  is
	\begin{align*}
	\Pi_0=&\{\hat\pi \in \Paths_\fin(N_0) \mid \last(\hat\pi) = c_0 \text{ and } c(\hat\pi) \leq R\} \cup \\
	&\{\hat\pi\ \in \Paths_\fin(N_0) \mid \last(\hat\pi) = \goalst \text{ and } c(\hat\pi) > 2R\}.
	\end{align*}
	Similarly in $N_1$ the following $p$-cause is $\pexpcost$-minimal:
	\begin{align*}
	\Pi_1=&\{\hat\pi \in \Paths_\fin(N_1) \mid \last(\hat\pi) = c_1\text{ and } c(\hat\pi) \leq R\} \cup \\
	&\{\hat\pi\ \in \Paths_\fin(N_1) \mid \last(\hat\pi) = \goalst \text{ and } c(\hat\pi) > 2R+1\}.
	\end{align*}
	Therefore we have
	\begin{align*}
	3 \cdot \pexpcost_{N_0}(\Pi_0) =& \sum_{\varphi \in \Paths_\fin(M), c(\varphi) \leq R} \Pr(\varphi)c(\varphi)\\
	&+\sum_{\varphi \in \Paths_\fin(M), c(\varphi)> R} \Pr(\varphi)\frac{c(\varphi)+R}{2},\\
	3 \cdot \pexpcost_{N_1}(\Pi_1) =& \sum_{\varphi \in \Paths_\fin(M), c(\varphi) \leq R} \Pr(\varphi)c(\varphi)\\
	&+\sum_{\varphi \in \Paths_\fin(M), c(\varphi) > R} \Pr(\varphi)\frac{c(\varphi)+R+1}{2}.
	\end{align*}
	We conclude that
	\begin{align*}
	&\pexpcost_{N_0}^{\min} - \pexpcost_{N_1}^{\min} \\
	=&\pexpcost_{N_0}(\Pi_0)-\pexpcost_{N_1}(\Pi_1) \\
	=&\frac{1}{6}\cdot\Pr_M(\{\varphi \in \Paths(M) \mid c(\varphi) > R\})\\
	=&\frac{1}{6}\cdot(1-\Pr_M(\{\varphi \in \Paths(M) \mid c(\varphi) \leq R\}))
	\end{align*}
	In the sequel we prove that we can use an oracle for the threshold problem for $\pexpcost^{\min}$ to compute the values
	\[\pexpcost_{N_0}^{\min} \qquad \text{and} \qquad  \pexpcost_{N_1}^{\min}.\] 
	This in turn allows us to compute $\frac{1}{6}\Pr_M(\{\hat\pi \in \Paths(M) \mid c(\hat\pi) \leq R\})$ and thus to decide the problem from \cite{HaaseK14}.
	
	In any acyclic Markov chain $K$, the following holds:
	We assume that the transition probabilities are encoded as fractions of coprime integers. Therefore we can compute the product of all denominators to get a number $L$ in polynomial time and with polynomially bounded encoding. The maximal weight $W$ of a path in $K$ can be computed in linear time and has a polynomially bounded encoding. Therefore the minimal $\pexpcost$ is an integer multiple of $1/L$ and there are at most $W \cdot L$ many different values for $\pexpcost$. This in particular applies to the value of the $\pexpcost$-optimal $p$-cause.
	
	Note that there are still exponentially many possible values for the minimal partial expected costs in both Markov chains $N_i$.
	A binary search over all possible values  with polynomially many applications of the threshold problem:
	\begin{center}
		Is $\pexpcost^{\min} \leq \vartheta$ for a rational $\vartheta \in \mathbb{Q}$?
	\end{center}
	is possible, nevertheless. We therefore can apply this binary search to find the exact value $\pexpcost_{N_i}(\Pi)$ for both DTMCs $N_i (i=0,1)$ by solving polynomially many instances of the threshold problem.
	This gives us a polynomial Turing reduction to the problem stated in \cite{HaaseK14}. \qed
\end{proof}

\maxcostSummary*
\begin{proof}
	We start with some preliminary considerations.
	If there exists a path $s_0 \ldots s_n$ entirely contained in $S \backslash S_p$ containing a cycle with positive weight, then $\maxcost$ of any $p$-cause is $\infty$:
	Consider a such a positive cycle reachable in $S \backslash S_p$.
	Then there are paths in $\goal$ which contain this cycle arbitrarily often.
	For any $p$-cause $\Pi$ almost all of these paths need a prefix in $\Pi$.
	Since no state in the positive cycle nor in the path from $s_0$ to the cycle is in $S_p$, such prefixes also contain the cycle arbitrarily often.
	This means these prefixes accumulate the positive weight of the cycle arbitrarily often.
	Therefore, all $p$-causes contain paths with arbitrary high weight.
	Thus, before optimizing $\Pi$ we check whether there are positive cycles reachable in the induced graph of $M$ on $S \backslash S_p$. 
	This can be done in polynomial time with the Bellman-Ford algorithm \cite{CormenLRS2009}.
	Henceforth we assume there are no such positive cycles.
	
	For (1) we show that for two $p$-causes $\Pi, \Phi$ with $\Pi \preceq \Phi$ we have $\maxcost(\Pi) \leq \maxcost(\Phi)$.
	Let $\hat\pi\in \Pi$ be arbitrary. Since $\hat\pi$ is a $p$-critical prefix (and $p>0$), there is a path $\pi \in \Paths_M(\lang) $ with $ \hat\pi \in \Pref(\pi)$.
	Since $\Phi$ is a $p$-cause, there exists $\hat \varphi \in \Phi \cap \Pref(\pi)$.
	The assumption $\Pi \preceq \Phi$ and the fact that both $\Pi$ and $\Phi$ are prefix-free then force $\hat\pi$ to be a prefix of $\hat\varphi$. Hence $c(\hat \pi) \leq c(\hat\varphi)$, and since $\hat\pi$ was arbitrary, it follows that $\maxcost(\Pi) \leq \maxcost(\Phi)$. This implies that $\Theta$ is $\maxcost$-minimal.
	
	Computing $\maxcost^{\min} = \maxcost(\Theta)$ for a non-negative weight function can be reduced to the computation of the  longest path in a modified version of $M$.
	There can be cycles with weight $0$ in $S \backslash S_p$, but in such cycles every state in the cycle has  weight $0$.
	Therefore we collapse such cycles completely without changing the value $\maxcost(\Theta)$.
	We further collapse the set $S_p$ into one absorbing state $f$.
	Computing $\maxcost(\Theta)$ now amounts to searching for a longest path from $s_0$ to $f$ in this modified weighted directed acyclic graph. 
	This can be done in linear time by finding a shortest path after multiplying all weights with $-1$.
	Therefore the problem can be solved in overall polynomial time \cite{CormenLRS2009}.
	
	\medskip
	For (2) we reduce the problem of finding a $\maxcost$-minimal $p$-cause to the solution of a max-cost reachability game as defined in \cite{BGHM2015}.
	Define the game arena $\mathcal{A}=(V, V_{Max}, V_{Min}, E)$ with
	\begin{align*}
		V= S \:\dot\cup\:\dot{S}_p, && V_{Max}=S, && V_{Min}=\dot{S}_p,
	\end{align*}
	where $\dot{S}_p$ is a copy of $S_p$. The copy of state $s\in S_p$ in $\dot{S}_p$ will be written as $\dot{s}$. There is an edge $(s,t) \in E$ between states $s$ and $t$ in $V$ if and only if one of the following conditions holds:
	\begin{enumerate}
		\item $s \in S ,t \in S \backslash S_p$, and $\mathbf{P}(s,t)>0$,
		\item $s \in S, t \in \dot{S}_p$, and for $u\in S_p$ with $t=\dot{u}$ we have $\mathbf{P}(s,u)>0$,
		\item $s \in \dot{S}_p, t \in S_p$, and $s = \dot{t}$, or
		\item $s \in \dot{S}_p$ and $t=\goalst$.
	\end{enumerate}
	We equip $\mathcal{A}$ with a weight function $w:V \to \mathbb{Q}$ by mirroring the weight function of $M$ in the following way:
	\begin{enumerate}
		\item $w(s) = c(s)$ for $s \in S \backslash S_p$, and
		\item $w(s) = 0$ for $s \in S_p$,
		\item $w(\dot{s}) = c(s)$ for $s \in S_p$,
	\end{enumerate}
	In \cite{BGHM2015} the authors define a min-cost reachability game.
	For our purposes we need the dual notion of a \emph{max}-cost reachability game, which is obtained by just changing the total payoff of a play avoiding the target set to be $-\infty$ instead of $+\infty$.
	The objective of player $Max$ with vertices $V_{Max}$ is then to maximize the total payoff of the play, while player $Min$ with vertices $V_{Min}$ wants to minimize the total payoff.
	By changing from min-cost to max-cost reachability games, the results of \cite{BGHM2015} concerning strategies for $Max$ and $Min$ are reversed.
	
	We consider the max-cost reachability game on $\mathcal{A}$ with target $\goalst$.
	Define $Val(s)$ as the total payoff if both sides play optimally if the play starts in $s$.
	We have $Val(s_0) = \infty$ if there is a positive cycle reachable in $S \backslash S_p$ by the above argumentation.
	In contrast we always have $Val(s_0) \neq -\infty$ since $\goalst$ is reachable and $\failst$ can be avoided by $Max$. We proceed to show that our reduction is correct.
	\begin{claim}
		\label{lem: MCR strategy p-cause}
		There is a $1$-$1$ correspondence between strategies $\sigma$ of $Min$ and $p$-causes in $M$ for $\goal$.
	\end{claim}
	\begin{proof}[of the claim]
		Let $\sigma$ be a strategy for $Min$ and consider the set of consistent plays starting in $s_0$ that we denote by $Plays_{s_0}(\mathcal{A},\sigma)$.
		Every play $\pi \in Plays_{s_0}(\mathcal{A}, \sigma)$ that reaches $\goalst$ corresponds to a $p$-critical prefix.
		To see this, omit all states contained in $ \dot{S}_p$ from the play.
		If $\pi$ reaches $\goalst$ and the last state before $\goalst$ is in $\dot{S}_p$ then omit $\goalst$ as well.
		The resulting path in $M$ is a $p$-critical prefix. 
		Let $\Pi \subseteq \Paths_\fin(M)$ be the set of $p$-critical prefixes obtained in this way from plays in $Plays_{s_0}(\mathcal{A},\sigma)$ reaching $\goalst$.
		
		To see that any path $\pi$ to $\goalst$ in $M$ has a prefix in $\Pi$, let $\tau$ be the strategy of player $Max$ that moves along the steps of $\pi$.
		In the resulting play, either player $Min$ moves to $\goalst$ from some state $s\in  \dot{S}_p$ according to $\sigma$ and the corresponding prefix of $\pi$ is in $\Pi$, or the play reaches $\goalst$ from a state not in $\dot{S}_p$ and hence the path $\pi$ itself belongs to $\Pi$.
		Since the strategy has to make decisions for every $s \in \dot{S}_p$, every path $\pi \in \Paths_M(\goal)$ has a prefix in $\Pi$.
		$\Pi$ is prefix-free since a violation of this property would correspond to a non-consistent play, since $Min$ can only choose the edge to $\goalst$ once.
		Therefore $\Pi$ is a $p$-cause in $M$ for $\goal$.
		
		Since the reverse of this construction follows along completely analogous lines, we omit it here. 
		\hfill$\blacksquare $
	\end{proof}
	
	\begin{claim}\label{cor: Val(s0)}
		We have $\maxcost^{\min}=Val(s_0)$.
	\end{claim}
	\begin{proof}[of the claim]
		Recall that the value $Val(s_0)$ is defined as the total payoff of the unique play $\pi \in Plays_{s_0}(\mathcal{A}, \sigma) \cap Plays_{s_0}(\mathcal{A}, \tau)$, where $\sigma$ is the optimal strategy of $Min$ and $\tau$ is the optimal strategy of $Max$.
		For each strategy $\sigma^\prime$ for $Min$, let $\Pi_{\sigma^\prime}$ be the corresponding $p$-cause as provided by the above claim. 
		The optimal strategy $\tau^\prime$ for $Max$ against $\sigma^\prime$ has to avoid $\failst$ and hence has to create a play that ends in $\goalst$.
		The total payoff of the induced play $\varphi$ is equal to the weight of the corresponding path $\hat\varphi\in \Pi_{\sigma^\prime}$ and the total payoff that $Max$ can achieve is precisely $\maxcost(\Pi_{\sigma^\prime})$.
		The optimal strategy $\sigma$ for $Min$ hence corresponds to a $\maxcost$-minimal $p$-cause $\Pi_{\sigma}$ and $\maxcost(\Pi_{\sigma})=Val(s_0)$.
				\hfill$\blacksquare $
	\end{proof}
	Now we apply the results for max-cost reachability games to $p$-causes.
	This means we can use \cite[Algorithm 1]{BGHM2015} to compute the value $Val(s)$ of the game for all states $s \in S$ in pseudo-polynomial time.
	This includes the value $Val(s_0)=\maxcost^{\min}$.
	From the values of the game starting at different states, an optimal memoryless strategy $\sigma$ for $Min$ can be derived by fixing 
	a successor $s^\prime$ for each state $s\in V_{Min}$ with $(s,s^\prime)\in E$ and 
	\[
	Val(s)=c(s^\prime)+ Val(s').
	\]
	Since the strategy is memoryless, we get a set $Q \subseteq \dot{S}_p$ for which $Min$ chooses the edge to $\goalst$.
	By the construction from before the $\maxcost$-minimal $p$-cause $\Pi$ obtained in this way is state-based.
	
	\medskip
	For (3) we note that the decision problem ``Is $Val(s_0) \leq 0$?''  is in $\np\cap\conp$ by a reduction to mean-payoff games, as shown in \cite{Chatterjee2017}. The reduction introduces an edge from $\goalst$ back to $s_0$ and removes the state $\failst$ from the game. We have that $Val(s_0) \leq 0$ in the original max-cost reachability game if and only if the value of the constructed mean-payoff game is at most $0$. The reason is that the value of the mean-payoff game is at most $0$ if there is a strategy for $Min$ such that $Max$ can neither force a positive cycle in the original max-cost reachability game nor reach $\goalst$ with positive weight in the original game.
	We adapt the construction to show that the decision problem ``is $\maxcost^{\min} \leq \vartheta$?'', for $\vartheta \in \mathbb{Q}$, is also in $\np\cap\conp$.
	This can be achieved by adding an additional vertex $s$ with weight $-\vartheta$ to $V_{Max}$, removing the edge between $\goalst$ and $s_0$ and adding two new edges, one from $\goalst$ to $s$ and one from $s$ to $s_0$.
	The value of the resulting mean-payoff game is then at most $0$ 
	if there is a strategy for $Min$ such that $Max$ can neither force a positive cycle in the original max-cost reachability game nor reach $\goalst$ with  weight above $\vartheta$ in the original game.
	\qed
\end{proof}

\instcost*

\begin{proof}
Recall that we now work with a set of error states $E$ instead of a single state $\goalst$ and a set of terminal safe states $F$ such that  $\Pr_{f}(\Diamond E)=0$ if and only if $f\in F$.  
	We first note that for an instantaneous weight function we reduce partial expected cost to expected cost and therefore only need to consider one case:
	Given a weight function $c\colon S\to\mathbb{Q}$, consider the weight function $c'$ obtained from $c$ by forcing $c'(f)=0$ for all $f\in F$.
	Then the partial expected cost with respect to $c_\inst$ equals the expected cost with respect to $c'_\inst$.
	
	For $\expcost_\inst$ we construct an MDP $\mathcal{N}=(S \:\dot\cup\: \dot{S}_p, \{pick, continue\}, s_0, \mathbf{P}')$, where $\dot{S}_p,$ is a disjoint copy of $S_p$ in which all states are terminal. The copy of state $s\in S_p$ in $\dot{S}_p$ will be written as $\dot{s}$. We define $\mathbf{P}'(s,continue,s')  = \mathbf{P}(s,s')$ for $s,s' \in S$, and $\mathbf{P}'(s,pick,\dot{s})  = 1$ for $s \in S_p$. The action $pick$ is not enabled states outside of $S_p$, and the action $continue$ is not enabled in $\dot{S}_p$. We define the weight function $c_\mathcal{N}:S\: \dot \cup \:\dot{S}_p \to \mathbb{Q}$  by
	\[c_\mathcal{N}(s)= \begin{cases}
		c(s)& \text{if }s \in F \cup E \\
		c(t)& \text{if $s \in \dot{S}_p$ and $\dot{t} = s$} \\
		0& \text{else.}
	\end{cases} \]

	The construction is illustrated
	in Figures \ref{fig: MDP construction instantaneous 1} and \ref{fig: MDP construction instantaneous 2}:
		Consider the DTMC $M^\prime$ depicted in \Cref{fig: MDP construction instantaneous 1}. The transition probabilities are omitted. It is enough to know $S_p^\prime= \{s_0,s_1 \}$ and $F^\prime= \{\failst \}$. The constructed MDP $\mathcal{N}^\prime$ can be seen in \Cref{fig: MDP construction instantaneous 2}, where the black edges are inherited from $M^\prime$, and the red edges are added transitions belonging to the action $pick$.
		\begin{figure}[t]
			\centering
			\begin{minipage}{0.45\textwidth}
				\centering
				\begin{tikzpicture}[->,>=stealth',shorten >=1pt,auto,node distance=0.5cm, semithick]
					
					\node[scale=0.8, state,accepting] (s0) {$s_0|+3$};
					\node[scale=0.8, state,accepting] (s1) [right = 0.8 of s0]{$s_1|+1$};
					\node[scale=0.8, state] (s2) [below = 0.5 of s1]{$s_2|-2$};
					\node[scale=0.8, state] (g) [right = 0.8 of s1]{$\goalst|+7$};
					\node[scale=0.8, state] (f) [right = 0.8 of s2]{$\failst|+4$};
					
					\draw[<-] (s0) -- ++(-0.55,0.55);
					\draw (s0) -- (s1) node[pos=0.5,scale=0.8] {};
					\draw (s0) -- (s2) node[pos=0.5,scale=0.8] {};
					\draw (s1) -- (s2) node[pos=0.5,scale=0.8] {};
					\draw (s1) -- (g) node[pos=0.5,scale=0.8] {};
					\draw (s2) -- (g) node[pos=0.5,scale=0.8] {};
					\draw (s2) -- (f) node[pos=0.5,scale=0.8] {};
					
				\end{tikzpicture}
				\caption{The DTMC $M^\prime$ with instantaneous weight} \label{fig: MDP construction instantaneous 1}
			\end{minipage}
			\hfill
			\begin{minipage}{0.45\textwidth}
				\centering
				\begin{tikzpicture}[->,>=stealth',shorten >=1pt,auto,node distance=0.5cm, semithick]
					
					\node[scale=0.8, state,accepting] (s0) {$s_0|+0$};
					\node[scale=0.8, state,accepting] (s1) [right = 0.8 of s0]{$s_1|+0$};
					\node[scale=0.8, state] (s2) [below = 0.5 of s1]{$s_2|+0$};
					\node[scale=0.8, state] (g) [right = 0.8 of s1]{$\goalst|+7$};
					\node[scale=0.8, state] (f) [right = 0.8 of s2]{$\failst|+4$};
					
					\node[scale=0.8, state,accepting] (s0copy) [above = 0.5 of s0]{$\dot{s}_0|+3$};
					\node[scale=0.8, state,accepting] (s1copy) [right = 0.8 of s0copy]{$\dot{s}_1|+1$};
					
					\draw[<-] (s0) -- ++(-0.55,0.55);
					\draw (s0) -- (s1) node[pos=0.5,scale=0.8] {};
					\draw (s0) -- (s2) node[pos=0.5,scale=0.8] {};
					\draw (s1) -- (s2) node[pos=0.5,scale=0.8] {};
					\draw (s1) -- (g) node[pos=0.5,scale=0.8] {};
					\draw (s2) -- (g) node[pos=0.5,scale=0.8] {};
					\draw (s2) -- (f) node[pos=0.5,scale=0.8] {};
					\draw[red] (s0) -- (s0copy) node[pos=0.5,scale=0.8] {};
					\draw[red] (s1) -- (s1copy) node[pos=0.5,scale=0.8] {};
					
				\end{tikzpicture}
				\caption{The MDP $\mathcal{N}^\prime$ emulating instantaneous weight} \label{fig: MDP construction instantaneous 2}
			\end{minipage}\hfill	
		\end{figure}
	
	For the constructed MDP $\mathcal{N}$ we consider the accumulated weight function.
	This emulates an instantaneous weight function for the random variable $\mathcal{X}_c$ of $\expcost_\inst$.
	A scheduler of this MDP corresponds to a $p$-cause for $M$ in the same way as established in \Cref{lem: scheduler trafo}.
	Therefore the problem of finding an $\expcost$-minimal $p$-cause for instantaneous weight $c:S \to \mathbb{Q}$ in $M$ for $\Goal$ is equivalent to finding a cost-minimal scheduler in $\mathcal{N}$ for $\Goal \cup F \cup \dot{S}_p$.
	This is again the stochastic shortest path problem for $\mathcal{N}$, which can be solved in polynomial time by \cite{BT91}.
	Since the SSP is solved by a memoryless scheduler, the $\expcost_\inst$-minimal $p$-cause is state-based.
	
	For the computation of the minimal value of $\maxcost_\inst$, we enumerate the set $S_p$ as $s^0,\dots,s^k$ where $k=|S_p|-1$ such that   $c(s^i) \leq c(s^j)$ for all $0\leq i < j \leq k$.
	Now we iteratively remove states in increasing order starting with $s^0$.
	After removing a state $s^i$, we check whether $E$ is reachable in the resulting Markov chain.
	If this is the case, we continue by removing the next state.
	 If $E$ is not reachable anymore, the set $S^i:=\{s^0,\dots,s^i\}$ induces a state-based $p$-cause $\Pi_{S^i}$.
	This follows from the fact that each path from the initial state to $E$ contains a state in $S^i$ and that $S^i\subseteq S_p$. Furthermore, $\maxcost_\inst (\Pi_{S^i})\leq c(s^i)$.
	Let $j$ be the largest number less than $i$ such that $c(s^j)<c(s^i)$. There is no $p$-cause in which all paths end in $\{s^0,\dots, s^j\}$ as $E$ was still reachable after removing these states from $M$.
	So, there is no $p$-cause $\Pi$ with $\maxcost_\inst (\Pi)<c(s^i)$. Therefore, $\Pi_{S^i}$ is indeed a $\maxcost_\inst$-minimal $p$-cause.
	Since $E\subseteq S_p$, the procedure terminates at the latest when the states in $E$ are removed. Hence the algorithm finds a state-based $\maxcost_\inst$-minimal $p$-cause in polynomial time.
	\qed
\end{proof}

\end{document}